\def\year{2020}\relax
\documentclass[letterpaper]{article}
\usepackage{style/aaai20}
\usepackage{times}
\usepackage{helvet}
\usepackage{courier}
\usepackage{url}
\usepackage{graphicx}
\usepackage{adjustbox}
\usepackage{algorithmicx}
\usepackage{amsmath}
\usepackage{amsthm}
\usepackage{comment}
\usepackage{diagbox}
\usepackage[shortlabels]{enumitem}
\usepackage{multirow}
\usepackage{thm-restate}
\usepackage[colorinlistoftodos,prependcaption,textsize=tiny]{todonotes}
\usepackage{xspace}

\newtheorem{definition}{Definition}
\newtheorem{example}{Example}
\newtheorem{lemma}{Lemma}
\newtheorem{theorem}{Theorem}
\newtheorem{proposition}{Fact}
\newtheorem{results}{Summary}

\newcommand{\Tuple}[1]{\ensuremath{\langle #1 \rangle}\xspace}

\newcommand{\FirstItem}{(a)\xspace}
\newcommand{\SecondItem}{(b)\xspace}
\newcommand{\ThirdItem}{(c)\xspace}
\newcommand{\Vector}[1]{\ensuremath{\vec{#1}}\xspace}
\newcommand{\Vx}{\Vector{x}}
\newcommand{\Vy}{\Vector{y}}

\newcommand{\Vw}{\Vector{w}}

\newcommand{\Vt}{\Vector{t}}
\newcommand{\Vu}{\Vector{u}}

\newcommand{\ShortDots}{\hbox to 1em{.\hss.\hss.}}

\newcommand{\DoubleExpTime}{\textsc{2ExpTime}\xspace}

\newcommand{\FormatEntitySet}[1]{\ensuremath{\mathbf{#1}}\xspace}
\newcommand{\Constants}{\FormatEntitySet{C}}
\newcommand{\Variables}{\FormatEntitySet{V}}
\newcommand{\Functions}{\FormatEntitySet{F}}
\newcommand{\Predicates}{\FormatEntitySet{P}}
\newcommand{\Terms}{\FormatEntitySet{T}}
\newcommand{\GroundTerms}{\FormatEntitySet{GT}}
\newcommand{\Entities}{\FormatEntitySet{E}}
\newcommand{\EntitiesIn}[2]{\ensuremath{#1_{#2}}\xspace}
\newcommand{\Arity}{\ensuremath{\textit{ar}}\xspace}
\newcommand{\Formula}{\ensuremath{\upsilon}\xspace}

\newcommand{\Fact}{\ensuremath{\varphi}\xspace}
\newcommand{\Rule}{\ensuremath{\rho}\xspace}
\newcommand{\RuleAux}{\ensuremath{\mu}\xspace}
\newcommand{\Body}{\ensuremath{\beta}\xspace}
\newcommand{\Head}{\ensuremath{\eta}\xspace}
\newcommand{\BCQ}{\ensuremath{\gamma}\xspace}
\newcommand{\FormatFormulaSet}[1]{\ensuremath{\mathcal{#1}}\xspace}
\newcommand{\FormulaSet}{\FormatFormulaSet{U}}
\newcommand{\FactSet}{\FormatFormulaSet{F}}
\newcommand{\FactSetAux}{\FormatFormulaSet{G}}
\newcommand{\AtomSet}{\FormatFormulaSet{A}}
\newcommand{\AtomSetAux}{\FormatFormulaSet{B}}

\newcommand{\RuleSet}{\FormatFormulaSet{R}}

\newcommand{\Ontology}{\FormatFormulaSet{O}}
\newcommand{\On}{\Ontology}
\newcommand{\Sk}[1]{\textit{sk}(#1)}
\newcommand{\SkApp}[2]{\ensuremath{\textit{sk}_{#1}(#2)}}
\newcommand{\Max}[1]{\ensuremath{\textit{max}(#1)}\xspace}
\newcommand{\Depth}[1]{\ensuremath{\textit{dep}(#1)}\xspace}

\newcommand{\Substitution}{\ensuremath{\sigma}\xspace}
\newcommand{\Subs}{\Substitution}
\newcommand{\SubstitutionAux}{\ensuremath{\tau}\xspace}
\newcommand{\SubsAux}{\SubstitutionAux}

\newcommand{\Chase}{\ensuremath{\mathcal{C}}\xspace}

\newcommand{\Rewriting}{\ensuremath{\pi}\xspace}

\newcommand{\Application}[3]{\ensuremath{#1_{#2, #3}}\xspace}

\newcommand{\EqualityAxiomatisationPredicate}{\textit{Eq}\xspace}
\newcommand{\EP}{\EqualityAxiomatisationPredicate}

\newcommand{\Sing}[1]{\ensuremath{\textsf{Sg}(#1)}}
\newcommand{\St}[1]{\ensuremath{\textsf{St}(#1)}}

\newcommand{\CriticalInstance}[1]{\ensuremath{\FactSet^\star_{#1}}\xspace}
\newcommand{\CI}[1]{\CriticalInstance{#1}}
\newcommand{\EMFASet}[1]{\ensuremath{\textsf{E}_{#1}}\xspace}
\newcommand{\MFA}{\text{MFA}\xspace}
\newcommand{\EMFA}{\text{EMFA}\xspace}
\newcommand{\RMFA}{\text{RMFA}\xspace}

\newcommand{\EPComplete}{\EP-complete\xspace}

\newcommand{\ClaimSing}{I}
\newcommand{\ClaimStandard}{II}
\newcommand{\ClaimEMFAGenerality}{III}
\newcommand{\ClaimEMFAPerformance}{IV}

\includecomment{tr}
\excludecomment{paper}

\title{Checking Chase Termination over Ontologies of Existential Rules with Equality}
\author{David Carral \\
Institute for Theoretical Computer Science \\
Technische Universit\"at Dresden, Germany \\
david.carral@tu-dresden.de
\And
Jacopo Urbani \\
Department of Computer Science \\
Vrije Universiteit Amsterdam, The Netherlands \\ 
acopo@cs.vu.nl}

\begin{document}

\maketitle

\begin{abstract}
The chase is a sound and complete algorithm for conjunctive query answering over ontologies of existential rules with equality.
To enable its effective use, we can apply acyclicity notions; that is, sufficient conditions that guarantee chase termination.
Unfortunately, most of these notions have only been defined for existential rule sets without equality.
A proposed solution to circumvent this issue is to treat equality as an ordinary predicate with an explicit axiomatisation.
We empirically show that this solution is not efficient in practice and propose an alternative approach.
More precisely, we show that, if the chase terminates for any equality axiomatisation of an ontology, then it terminates for the original ontology (which may contain equality).
Therefore, one can apply existing acyclicity notions to check chase termination over an axiomatisation of an ontology and then use the original ontology for reasoning.
We show that, in practice, doing so results in a more efficient reasoning procedure.
Furthermore, we present equality model-faithful acyclicity, a general acyclicity notion that can be directly applied to ontologies with equality.
\end{abstract}

\section{Introduction}
\label{section:introduction}

Answering conjunctive queries (CQs) over ontologies of existential rules with equality is a relevant reasoning task, which is undecidable \cite{cq-undecidable-erules}.
One approach to solve it in some cases is to use the \emph{chase} \cite{old-school-chase}---a forward-chaining algorithm, which is sound and complete but may not terminate.
Despite the fact that checking chase termination is undecidable \cite{undecidable-chase-termination,anatomy-chase}, we can apply \emph{acyclity notions}---sufficient conditions that guarantee termination---to enable the effective use of the chase for a large subset of real-world ontologies \cite{mfa,rmfa}.

Acyclicity notions have been widely researched and many such criteria have been developed \cite{wa,swa,graph-rule-dependencies,ja,rca,rmfa,k-safe}.
Alas, some of the most general notions, such as model-faithful acyclicity (MFA) \cite{mfa}, are only defined for existential rule sets without equality.
This restriction limits their usefulness, since equality is a prevalent feature (for instance, equality is used in $\mathtt{\sim}34\%$ of the logical theories in the Oxford Ontology Repository\footnote{\url{www.cs.ox.ac.uk/isg/ontologies/}}).

A proposed solution to enable the use of existing acyclicity notions over ontologies with equality is to treat equality as an ordinary predicate with an explicit axiomatisation (see Sections 2.1 and 5 of \cite{mfa}).
Intuitively, an axiomatisation of a rule set \RuleSet is another rule set that does not contain equality and can be exploited to solve CQ answering over ontologies with the rule set \RuleSet.
More precisely, using axiomatisations, we can solve CQ answering over an ontology $\On = \Tuple{\RuleSet, \FactSet}$, where \RuleSet is an existential rule set (possibly containing equality) and \FactSet is a fact set, by implementing the following step-by-step approach:
\begin{enumerate}
\item Compute some equality axiomatisation $\RuleSet'$ of $\RuleSet$. \label{step:axiomatisation}
\item Verify whether $\RuleSet'$ is acyclic with respect to some acyclicity notion (e.g., \MFA). \label{step:acyclicity}
If this is the case, then the chase of $\RuleSet'$ terminates; that is, for any given fact set $\FactSet'$, the chase terminates on input $\Tuple{\RuleSet', \FactSet'}$.
\item Apply the chase on $\Tuple{\RuleSet', \FactSet}$, and use the resulting output fact set to solve CQ answering over \On. \label{step:chase}
\end{enumerate}
Note that, since $\RuleSet'$ is an axiomatisation of \RuleSet, $\RuleSet'$ is equality-free and hence, we can check if this rule set is \MFA in Step~\ref{step:acyclicity}.

The application of the above step-by-step approach to real-world ontologies is somewhat problematic.
For instance, the use of the \emph{standard axiomatisation} in Step~\ref{step:axiomatisation} often causes the \MFA check applied in Step~\ref{step:acyclicity} to fail \cite{mfa}.
As shown in this paper, the use of other axiomatisation techniques in Step~\ref{step:axiomatisation}, such as \emph{singularisation} \cite{swa}, fixes this issue to a large extent.
Unfortunately, computing the chase of an ontology that features some singularisation of a rule set---as required in Step~\ref{step:chase}---is not efficient in practice.
The fact that the use of axiomatisations leads to poor performance has previously been shown for the standard axiomatisation \cite{rewriting}; we show that it is also the case when singularisation is applied.

To address these issues, we show that, if the chase of any equality axiomatisation of \RuleSet terminates, then so does the chase of the rule set \RuleSet.
Hence, we can replace Step~\ref{step:chase} in the above step-by-step procedure with the following alternative:
\begin{enumerate}
\setcounter{enumi}{3}
\item Compute the chase on input \On and use the resulting output fact set to solve CQ answering over \On. \label{step:chase-alternative}
\end{enumerate}
Implementing Step~\ref{step:chase-alternative} instead of Step~\ref{step:chase} enables the use of \emph{rewriting} to deal with equality when computing the chase; a technique that has already been proven more efficient than the use of axiomatisations in practice \cite{rewriting}.

Still, there is yet another practical problem.
Namely, checking if the singularisation of a rule set is \MFA---as required in Step~\ref{step:acyclicity}---is somewhat inefficient for many real-world rule sets.
To solve this issue, we present \emph{equality model-faithful acyclicity} (\EMFA), a very general acyclicity notion based on \MFA that can be directly applied to rule sets with equality.
By applying this notion directly, we altogether remove the need for using equality axiomatisations.

Our contributions are as follows:
we provably show that, if the chase of the standard axiomatisation or any singularisation of a rule set \RuleSet terminates, then so does the chase of \RuleSet; we show that the converses of the previous implications do not hold; and we define \EMFA, we study the complexity of checking \EMFA membership and reasoning over \EMFA ontologies, and we compare the expressivity of this notion with that of \MFA.
Moreover, we empirically show that
\begin{enumerate}[a.]
\item[\ClaimSing.] computing the chase of an ontology featuring some singularisation of a rule set is not efficient in practice,
\item[\ClaimStandard.] the standard equality axiomatisation of a large subset of real-world rule sets is not \MFA,
\item[\ClaimEMFAGenerality.] \EMFA is as general as ``\MFA{} plus singularisation'', and
\item[\ClaimEMFAPerformance.] checking if a rule set is \EMFA is more efficient than checking if it is ``\MFA{} plus singularisation''.
\end{enumerate}
\begin{paper}
Our treatment is fully self-contained, but a technical report with further information can be consulted if desired \cite{technical-report}.
\end{paper}
\begin{tr}
Formal proofs for all technical results are included in the appendix of this document.
\end{tr}

\section{Preliminaries}
\label{section:preliminaries}

\subsection{Syntax and Semantics}

We consider a signature based on mutually disjoint, finite sets of \emph{constants} \Constants, \emph{function symbols} \Functions, \emph{variables} \Variables, and \emph{predicates} \Predicates.
Every entity $e \in \Functions \cup \Predicates$ is associated with some arity $\Arity(e) \geq 1$.
The set \Predicates includes the special binary predicate $\mathop{\approx}$, which is referred to as the \emph{equality predicate} or simply as \emph{equality}.
The set of \emph{terms} $\Terms$ is the minimal superset of \Constants and \Variables such that, for all $f \in \Functions$ and all $t_1, \ldots, t_{\Arity(f)} \in \Terms$, we have that $f(t_1, \ldots, t_{\Arity(f)}) \in \Terms$.
The set of \emph{ground terms} $\GroundTerms$ is the set of all terms without syntactic occurrences of a variable.
For a term $t$, let $\Depth{t} = 1$ if $t \in \Constants \cup \Variables$, and $\Depth{t} = \Max{\Depth{t_1}, \ldots, \Depth{t_n}} + 1$ if $t$ is of the form $f(t_1, \ldots, t_n)$.
Given an entity set \Entities and a formula or set thereof \FormulaSet, we write \EntitiesIn{\Entities}{\FormulaSet} to denote the set that contains all of the elements in \Entities that occur in \FormulaSet.
We abbreviate lists of terms $t_1, \ldots, t_n$ as $\vec{t}$ and treat these as sets when order is irrelevant.
An \emph{atom} is a formula $P(\vec{t})$ with $P \in \Predicates$, $\vec{t} \in \Terms$, and $\Arity(P) = \vert \vec{t} \vert$.
As customary, we write $t \approx u$ instead of $\mathop{\approx}(t, u)$ to denote atoms defined over equality.

A \emph{fact} is an atom $P(\Vt)$ with $\Vt \in \Constants$.
For a formula \Formula and a list of variables \Vx, we write $\Formula[\Vx]$ to indicate that \Vx is the set of all free variables occurring in \Formula (i.e., the set of all variables that are not quantified in \Formula).
An \emph{(existential) rule} is a function- and constant-free first-order logic (FOL) formula of either of the following forms.
\begin{align}
\forall \Vx, \Vy . \big(\Body[\Vx, \Vy] &\to \exists \Vw . \Head[\Vx, \Vw]\big) \label{rule:tgd} \\
\forall \Vx . \big(\Body[\Vx] &\to x \approx y\big) \label{rule:egd}
\end{align}
In the above, \Vx, \Vy, and \Vw are pairwise disjoint lists of variables; \Vx is non-empty; \Body and \Head are non-empty conjunctions of atoms without equality; and $x, y \in \Vx$.
The \emph{body} (resp. \emph{head}) of a rule is the conjunction of atoms to the left (resp. right) of its implication symbol.
We omit universal quantifiers when writing rules and treat conjunctions of atoms, such as \Body and \Head above, as atom sets.
We refer to rules of the form \eqref{rule:tgd} and \eqref{rule:egd} as \emph{tuple generating dependencies} (TGDs) and \emph{equality generating dependencies} (EGDs), respectively.

A \emph{boolean conjunctive query} (BCQ) is a function-free FOL formula $\BCQ = \exists \Vx . \Body[\Vx]$ with \Body a non-empty conjunction of atoms that, without loss of generality, does not contain constants or equality.
We refer to \Body as the \emph{body} of \BCQ.
Since CQ answering can be reduced to BCQ entailment, we confine our attention to the latter without loss of generality.

We consider finite rule sets \RuleSet, where we assume without loss of generality that existentially quantified variables do not reoccur across different rules ($\dagger$).
An \emph{ontology} \On is a tuple \Tuple{\RuleSet, \FactSet} with \RuleSet a rule set and \FactSet a fact set.
Without loss of generality, we assume that, for an ontology \Tuple{\RuleSet, \FactSet}, the set \FactSet is equality-free and $\EntitiesIn{\Predicates}{\FactSet} \subseteq \EntitiesIn{\Predicates}{\RuleSet}$.

For an ontology \On and a BCQ \BCQ, we write $\On \models \BCQ$ to indicate that \On entails \BCQ under FOL semantics.
That is, to indicate that $\bigwedge_{\Rule \in \RuleSet} \Rule \wedge \bigwedge_{\Fact \in \FactSet} \Fact$ entails \BCQ.

\subsection{The Non-Oblivious Chase Algorithm}

We present the \emph{(non-oblivious) chase} \cite{wa}---a chase variant that expands existential quantifiers only if necessary, and merges terms to comply with the semantics of equality.
Unlike Fagin et al. \shortcite{wa}, we do not contemplate the \emph{unique name assumption} and introduce Skolem functional terms instead of ``unlabelled'' nulls to satisfy existential restrictions.
The use of ``labelled'' Skolem terms simplifies some the formal arguments presented in the following sections (e.g., see the proofs of Theorems~\ref{theorem:standard}, \ref{theorem:singularisation}, and \ref{theorem:reasoning-membership}).

\begin{definition}[Skolemisation]
The \emph{skolemisation} \Sk{\Rule} of a TGD \Rule of the form \eqref{rule:tgd} is the formula $\Body \to \Sk{\Head}$ where $\Sk{\Head}$ is the conjunction of atoms obtained from \Head by replacing each variable $w \in \Vw$ by the Skolem term $f_w(\Vx)$, with $f_w$ a fresh Skolem function symbol of arity $\vert \Vx \vert$ unique for $w$.
\end{definition}

By ($\dagger$), an existentially quantified variable $w$ in some rule set \RuleSet may occur in at most one rule $\Rule \in \RuleSet$ and hence, the function symbol $f_w$ is unique for both $w$ and $\Rule$ in \Sk{\RuleSet}.


We write $[t_1 / u_1, \ldots, t_n / u_n]$ to denote the function over \Terms mapping $t_i$ to $u_i$ for all $i \in \{1, \ldots, n\}$.
Given such a function $\Subs$ and a formula \Formula, let $\Formula\Subs$ be the formula that results from replacing every occurrence of a term $t$ as a predicate argument in an atom in \Formula with $\Subs(t)$ if the latter is defined  (e.g., $P(t, f(t))[t / u] = P(u, f(t))$).
A \emph{rewriting} (resp. \emph{substitution}) is a functions from \GroundTerms (resp. \Variables) to \GroundTerms.

To guide the renaming that results from applying EGDs, we introduce a strict total order $\prec$ defined over the set of terms such that $t \prec u$ for all $t, u \in \Terms$ with $\Depth{t} < \Depth{u}$.

\begin{definition}[Rule Applicability]
\label{definition:applicability}
Consider a rule \Rule, a substitution \Subs, and an atom set \AtomSet.
Then, the tuple \Tuple{\Rule, \Subs} is \emph{applicable} to \AtomSet if all of the following hold.
\begin{itemize}
\item \Subs is defined and undefined for all of the universally and existentially quantified variables in \Rule, respectively.
\item $\Body\Subs \subseteq \AtomSet$ with \Body the body of \Rule.
\item If \Rule is of the form \eqref{rule:tgd}, then $\Head\SubsAux \not \subseteq \AtomSet$ for all $\SubsAux \supseteq \Subs$.
\item If \Rule is of the form \eqref{rule:egd}, then $\Subs(x) \neq \Subs(y)$.
\end{itemize}

If \Tuple{\Rule, \Subs} is applicable to \AtomSet, then the \emph{application} of \Tuple{\Rule, \Subs} on \AtomSet, denoted with $\Application{\AtomSet}{\Rule}{\Subs}$, is the atom set defined as follows.
\begin{itemize}
\item If \Rule is of the form \eqref{rule:tgd}, then $\Application{\AtomSet}{\Rule}{\Subs} = \AtomSet \cup \SkApp{\Subs}{\Head}$ where \SkApp{\Subs}{\Head} is the formula that results from replacing every syntactic occurrence of a variable $x$ in \Sk{\Head} with $\Subs(x)$.
\item If \Rule is of the form \eqref{rule:egd}; then $\Application{\AtomSet}{\Rule}{\Subs} = \AtomSet [\Subs(y) / \Subs(x)]$ if $\Subs(x) \prec \Subs(y)$, and $\Application{\AtomSet}{\Rule}{\Subs} = \AtomSet [\Subs(x) / \Subs(y)]$ otherwise.
\end{itemize}
\end{definition}

The case distinction in the last item in the previous definition ensures that we merge ``deeper'' terms into ``shallower'' ones when applying EGDs.
This strategy simplifies some of our formal arguments (e.g., see the proof of Theorem~\ref{theorem:bcq-membership}), but it is not necessary to define a correct procedure.

\begin{definition}[The Chase Algorithm]
\label{definition:chase}
A \emph{chase sequence} of an ontology $\On = \Tuple{\RuleSet, \FactSet}$ is a (possibly finite) sequence $\AtomSet_0, \AtomSet_1, \ldots$ of atom sets such that the following hold.
\begin{itemize}
\item $\AtomSet_0 = \FactSet$
\item For all $i \geq 1$, there is some rule $\Rule \in \RuleSet$ and some substitution \Subs such that $\AtomSet_i$ is the application of \Tuple{\Rule, \Subs} on $\AtomSet_{i-1}$.
\item For all rules $\Rule \in \RuleSet$ and substitutions \Subs, there is a $k \geq 0$ such that, for all $i \geq k$, the tuple \Tuple{\Rule, \Subs} is not applicable to $\AtomSet_i$ (i.e., fairness).
\end{itemize}
A \emph{chase} of \On is the atom set that results from taking the union of all atom sets in some chase sequence of \On.

The chase of an ontology \On \emph{terminates} if all of the chases of \On are finite; that is, if all chase sequences of \On are finite.
The chase of a rule set \RuleSet \emph{terminates} if, for all fact sets \FactSet, the chase of \Tuple{\RuleSet, \FactSet} terminates.
An atom set is a \emph{chase step} of \On if it occurs in a chase sequence of \On.
\end{definition}

An ontology may admit infinitely many different chases.
Nevertheless, it is well-known that any chase of an ontology is a universal model for this logical theory; i.e., a model that can be homomorphically embedded into any other model.
Therefore, this structure can be directly used to solve BCQ entailment \cite{chase-revisited}.
\begin{proposition}
An ontology entails a BCQ iff any chase of this ontology entails this BCQ.
\end{proposition}

This result holds because, by our definition, rules and BCQs do not contain constants.
If we want to lift this restriction, we would need to modify the definition of the chase so it takes into account the rewriting of terms that occurs when EGDs are applied as it is done in \cite{rewriting}.

\subsection{Handling Equality via Axiomatisation}
\label{section:axiomatisation}

We present two well-known techniques to compute equality axiomatisations; namely, the \emph{standard axiomatisation} (see Section 2.1 in \cite{mfa}) and \emph{singularisation} (see Section 5 in \cite{swa}).
In the definition of these, we replace $\approx$ with the fresh predicate \EP to clarify that these two predicates are to be treated differently.

\begin{definition}
\label{definition:standard}
The \emph{standard axiomatisation} of a rule set \RuleSet, denoted with $\St{\RuleSet}$, is the TGD set that contains all of the TGDs in \RuleSet, the TGD $\Body \to \EP(x, y)$ for every EGD $\Body \to x \approx y \in \RuleSet$, an instance of \eqref{rule:reflexivity} for every $P \in \EntitiesIn{\Predicates}{\RuleSet} \setminus \{\approx\}$, the TGDs \eqref{rule:symmetry} and \eqref{rule:transitivity}, and an instance of \eqref{rule:replacement} for every $P \in \EntitiesIn{\Predicates}{\RuleSet} \setminus \{\approx\}$ and every $i \in \{1, \ldots, \Arity(P)\}$.
\begin{align}
P(\Vx_{\Arity(P)}) &\to \bigwedge\nolimits_{1 \leq i \leq \Arity(P)} \EP(x_i, x_i) \label{rule:reflexivity} \\
\EP(x, y) &\to \EP(y, x) \label{rule:symmetry} \\
\EP(x, y) \wedge \EP(y, z) &\to \EP(x, z) \label{rule:transitivity} \\
P(\Vx_{\Arity(P)}) \wedge \EP(x_i, x_i') &\to \big(P(\Vx_{\Arity(P)})[x_i / x_i'] \big) \label{rule:replacement}
\end{align}
In the above, $\Vx_k = x_1, \ldots, x_k$ for all $k \geq 1$.

The \emph{standard axiomatisation} of an ontology $\On = \Tuple{\RuleSet, \FactSet}$ is the ontology $\St{\On} = \Tuple{\St{\RuleSet}, \FactSet}$.
\end{definition}

The standard axiomatisation of an ontology \On can be directly used to solve BCQ entailment over \On.
\begin{proposition}
\label{proposition:standard}
For an ontology \On and a BCQ \BCQ that does not contain the predicate\EP, we have $\On \models \BCQ$ iff $\St{\On} \models \BCQ$.
\end{proposition}

By applying this result, we can decide BCQ entailment over \On with algorithms (or implementations!) that do not require a special treatment to deal with equality.

\begin{definition}
\label{definition:singularisation}
Consider a conjunction of atoms \Body.
A \emph{singularisation} of \Body is a conjunction of atoms that is constructed by applying the following step-by-step transformation to \Body.
\begin{enumerate}
\item For every $x \in \EntitiesIn{\Variables}{\Body}$, choose some $k_x \in \{1, \ldots, n\}$ with $n$ the number of occurrences of $x$ in \Body.
\item Simultaneously, replace the $i$-th occurrence of every variable $x$ in \Body with a fresh variable $x_i$ if $i \neq k_x$.
\item Add the atom $\EP(x, x_i)$ for every fresh variable $x_i$ introduced in the above step.
\end{enumerate}
Let $\Sing{\Body}$ be the set of all singularisations of \Body.

The \emph{singularisations} of a TGD $\Rule = \Body \to \exists \Vw . \Head$, an EGD $\RuleAux = \Body \to x \approx y$, a rule set \RuleSet, an ontology $\On = \Tuple{\RuleSet, \FactSet}$, and a BCQ $\BCQ = \exists \Vw . \Body$ are defined as follows, respectively.
\begin{itemize}
\item $\Sing{\Rule} = \{\Body' \to \exists \Vw . \Head \mid \Body' \in \Sing{\Body}\}$
\item $\Sing{\RuleAux} = \{\Body' \to \EP(x, y) \mid \Body' \in \Sing{\Body}\}$
\item Let $\Sing{\RuleSet}$ be the set of all TGD sets that contain an instance of \eqref{rule:reflexivity} for every $P \in \EntitiesIn{\Predicates}{\RuleSet} \setminus \{\approx\}$, the TGDs \eqref{rule:symmetry} and \eqref{rule:transitivity}, and (exactly) one TGD in \Sing{\Rule} for each $\Rule \in \RuleSet$.
\item $\Sing{\On} = \{\Tuple{\RuleSet', \FactSet} \mid \RuleSet' \in \Sing{\RuleSet}\}$
\item $\Sing{\BCQ} = \{\exists \Vx . \Body' \mid \Body'[\Vx] \in \Sing{\Body}\}$
\end{itemize}
\end{definition}

The use of singularisation may yield many axiomatisations for a single ontology \On, any of which can be used to solve BCQ over \On \cite{swa}.

\begin{proposition}
\label{proposition:singularisation}
For an ontology \On and a BCQ \BCQ that does not contain the predicate \EP, we have $\On \models \BCQ$ iff $\On' \models \BCQ'$ for any $\On' \in \Sing{\On}$ and $\BCQ' \in \Sing{\BCQ}$.
\end{proposition}

Even though the use of equality axiomatisations does not affect the results of BCQ entailment (see Facts~\ref{proposition:standard} and \ref{proposition:singularisation}), it may influence chase termination.
For instance, in the following section, we show examples of rule sets \RuleSet such that \RuleSet is terminating with respect to the chase, and \Sing{\RuleSet} or some set in \St{\RuleSet} are not (see Theorems~\ref{example:1} and \ref{example:2}).

\section{Chase Termination and Axiomatisations}
\label{section:chase-termination}

In this section, we show that if the chase terminates for the standard axiomatisation or any singularisation of a given rule set \RuleSet, then the chase of \RuleSet also terminates (see Theorems~\ref{theorem:standard} and \ref{theorem:singularisation}).
Moreover, we provide some results stating that the converses of these implications do not hold (see Theorems~\ref{example:1} and \ref{example:2}).
Prior to showing these results, we introduce some preliminary definitions.

An atom set \AtomSet \emph{satisfies} a rule \Rule iff there is no substitution \Subs such that \Tuple{\Rule, \Subs} is applicable to \AtomSet.

\begin{definition}
\label{definition:epcompleteness}
An atom set \AtomSet is \emph{\EPComplete} iff it satisfies the rules \eqref{rule:symmetry} and \eqref{rule:transitivity} introduced in Definition~\ref{definition:standard}, and $\EP(t, t) \in \AtomSet$ for every $t \in \EntitiesIn{\Terms}{\AtomSet}$.
\end{definition}

\begin{definition}
\label{definition:rewriting-standard}
Consider an \EPComplete atom set \AtomSet.
Then, let $\Rewriting_\AtomSet$ be the rewriting that maps every $t \in \EntitiesIn{\Terms}{\AtomSet}$ into the term $\Rewriting_\AtomSet(t) \in \EntitiesIn{\Terms}{\AtomSet}$ such that
\begin{itemize}
\item $\EP(t, \Rewriting_\AtomSet(t)) \in \AtomSet$ and
\item for all $u \in \EntitiesIn{\Terms}{\AtomSet}$ with $u \neq \Rewriting_\AtomSet(t)$ and $\EP(t, u) \in \AtomSet$, we have that $\Rewriting_\AtomSet(t) \prec u$.
\end{itemize}
Furthermore, let $[\AtomSet] = \AtomSet\Rewriting_\AtomSet \setminus \{\EP(t, u) \mid t, u \in \Terms\}$.
\end{definition}

\begin{restatable}{theorem}{theoremstandard}\label{theorem:standard}
The chase of a rule set \RuleSet terminates if the chase of the TGD set \St{\RuleSet} terminates.
\end{restatable}
\begin{proof}[Proof Sketch]
Consider some fact set \FactSet and the ontology $\On = \Tuple{\RuleSet, \FactSet}$.
Theorem~\ref{theorem:standard} follows if chase termination of $\St{\On} = \Tuple{\St{\RuleSet}, \FactSet}$ implies chase termination of $\On$.
Note that, if \St{\On} terminates, then the number of chase steps of this ontology is necessarily finite.
In turn, this claim follows if, for every chase step $\AtomSet$ of \On, there is some \EPComplete chase step $\AtomSetAux$ of \St{\On} such that $[\AtomSetAux] = \AtomSet$.
This implication can be proven via induction.
\end{proof}

A proof of Theorem~\ref{theorem:standard} for the standard chase was presented
in Theorem 4.2 \cite{anatomy-chase}. For completeness, we present a proof for
the non-oblivious chase (which is the chase variant considered in this
paper) in the extended version of this paper.

%

We can show with a counter-example that the converse of Theorem~\ref{theorem:standard} does not hold.

\begin{theorem}
\label{example:1}
Even if the chase of a rule set \RuleSet terminates, the chase of the TGD set \St{\RuleSet} may not.
\end{theorem}
\begin{proof}
The chase of the rule set $\RuleSet = \{\eqref{example:rule1}, \eqref{example:rule2}\}$ does terminate whereas the chase of \St{\RuleSet} does not.
\begin{align}
A(x) &\to \exists w . R(x, w) \wedge B(w) \label{example:rule1} \\
R(x, y) \wedge R(x, z) &\to y \approx z \label{example:rule2}
\end{align}
For instance, the chase of the ontology $\St{\On} = \Tuple{\St{\RuleSet}, \FactSet}$ where $\FactSet$ is the fact set $\{A(a), R(a, a)\}$ admits infinite chase sequences.
Namely, \St{\On} admits (fair and infinite) chase sequences that contain all of the atom sets in the sequence $\FactSet_1 \cup \FactSet, \FactSetAux_1, \FactSet_2, \FactSetAux_2, \FactSet_3, \FactSetAux_3, \ldots$ where $\FactSetAux_i$ is the set of all atoms that can be defined using the predicates $A$, $B$, $R$, and $\EP$, and the terms in $\{f^0_w(a), \ldots, f^{i-1}_w(a)\}$ for all $i \geq 1$;
\begin{align*}
\FactSet_1 = \{	&R(a, f_w^1(a)), B(f_w^1(a)), \EP(a, f_w^1(a)), A(f_w^1(a))\}; \\
\FactSet_i = \{	&R(f_w^{i-1}(a), f_w^i(a)), B(f_w^i(a)), \EP(f_w^{i-1}(a), a), \\
			&R(a, f_w^i(a)), \EP(a, f_w^i(a)), A(f_w^i(a))\} \cup \FactSetAux_{i-1}
\end{align*}
for all $i \geq 2$; and $f^0_w(a) = a$, $f^1_w(a) = f_w(a)$, $f^2_w(a) = f_w(f_w(a))$, and so on.
\end{proof}


The fact that the chase does not terminate for the standard axiomatisation of a rule set as simple as the one described above illustrates why the use of this axiomatisation technique often causes the acyclicity tests to fail.
We empirically verify this insight in the evaluation section.

As per Definition~\ref{definition:singularisation}, a rule set may admit many different singularisations.
If the chase of at least one of these terminates, then so does the chase of the rule set itself.

\begin{restatable}{theorem}{theoremsingularisation}\label{theorem:singularisation}
The chase of a rule set \RuleSet terminates if the chase of some TGD set in \Sing{\RuleSet} terminates.
\end{restatable}
\begin{proof}[Proof Sketch]
Assume that there is some rule set $\RuleSet' \in \Sing{\RuleSet}$ for which the chase terminates.
Then, let \FactSet be some fact set, let $\On = \Tuple{\RuleSet, \FactSet}$ and $\On' = \Tuple{\RuleSet', \FactSet}$, and let $\AtomSet_0, \AtomSet_1, \ldots$ be a chase sequence of \On.
We can show via induction that there is a sequence of atom sets $\AtomSetAux_0, \AtomSetAux_1, \ldots$ and a sequence of rewritings $\Rewriting_0, \Rewriting_1, \ldots$ such that, for all $i \geq 0$,
\begin{enumerate}
\item $\AtomSetAux_i$ is a chase step of $\On'$,
\item $\AtomSetAux_i\Rewriting_i \setminus \{\EP(t, u) \mid t, u \in \Terms\} = \AtomSet_i$, and
\item $\EP(t, u) \in \AtomSetAux_i$ for all $t, u \in \Terms$ in $\AtomSetAux_i$ with $\Rewriting_i(t) = \Rewriting_i(u)$.
\end{enumerate}
Since (1), we conclude that the sequence $\AtomSetAux_0, \AtomSetAux_1, \ldots$ is finite and hence, the sequence $\AtomSet_0, \AtomSet_1, \ldots$ is also finite by (2).
In turn, this implies the chase of \RuleSet terminates.
Item (3) is an auxiliary claim used to structure the induction.
\end{proof}


We can show with a counter-example that the converse of Theorem~\ref{theorem:singularisation} does not hold.

\begin{theorem}
\label{example:2}
Even if the chase of a rule set \RuleSet terminates, the chase of every TGD set in $\Sing{\RuleSet}$ may not.
\end{theorem}

\begin{proof}
The chase of the rule $\RuleSet = \{\eqref{rule:theo2-1}, \eqref{rule:theo2-2}, \eqref{rule:theo2-3}\}$ does terminate whereas the chase every $\RuleSet' \in \Sing{\RuleSet}$ does not.
\begin{align}
B(x) \wedge C(x) &\to \exists y . R(x, y) \wedge B(y) \label{rule:theo2-1}\\
B(x) \wedge C(x) &\to \exists z . R(x, z) \wedge C(z) \label{rule:theo2-2}\\
R(x, y) &\to x \approx y \label{rule:theo2-3}
\end{align}
The chase of rule set \RuleSet does terminate as we have that, for any given fact set \FactSet, the chase of \Tuple{\RuleSet, \FactSet} may only contain terms of depth at most 1.

None of the four different rule sets in \Sing{\RuleSet} does terminate.
More precisely, for some $\RuleSet' \in \Sing{\RuleSet}$, the chase of \Tuple{\RuleSet', \{B(a), C(a)\}} does not terminate.
\end{proof}

\section{Extending \MFA with Equality}
\label{section:equality-mfa}

In this section, we present \emph{equality model-faithful acyclicity} (\EMFA), an acyclicity notion based on model-faithful acyclicity (MFA) \cite{mfa} that can be directly applied to existential rule sets with equality.

\begin{definition}
A term $t$ is \emph{cyclic} if it is of the form $f(\Vu)$ and the function symbol $f$ occurs in some term in \Vu.
\end{definition}

\begin{definition}[MFA/\EMFA]
\label{definition:emfa}
For a rule set \RuleSet, let $\EMFASet{\RuleSet}$ be the minimal atom set that satisfies the following.
\begin{itemize}
\item $\CI{\RuleSet} \subseteq \EMFASet{\RuleSet}$ with \CI{\RuleSet} the \emph{critical instance} for \RuleSet; that is, the set of all facts that can be defined using the predicates in \EntitiesIn{\Predicates}{\RuleSet} and the special constant $\star$.
\item If there is a TGD $\Body \to \exists \Vw. \Head \in \RuleSet$ and a substitution $\Subs$ with $\Body\Subs \subseteq \EMFASet{\RuleSet}$, then $\SkApp{\Subs}{\Head} \subseteq \EMFASet{\RuleSet}$.
\item If there is an EGD $\Body \to x \approx y \in \RuleSet$ and a substitution $\Subs$ with $\Body\Subs \subseteq \EMFASet{\RuleSet}$; then $\EMFASet{\RuleSet}[\Subs(y) / \Subs(x)] \subseteq \EMFASet{\RuleSet}$ if $\Depth{\Subs(x)} \leq \Depth{\Subs(y)}$, and $\EMFASet{\RuleSet}[\Subs(x) / \Subs(y)] \subseteq \EMFASet{\RuleSet}$ if $\Depth{\Subs(y)} \leq \Depth{\Subs(x)}$.
\end{itemize}
A rule set \RuleSet is \emph{\EMFA} if \EMFASet{\RuleSet} does not contain any cyclic terms.
A TGD set is \emph{\MFA} if it is \EMFA.
An ontology $\Tuple{\RuleSet, \FactSet}$ is \emph{\MFA} (resp. \emph{\EMFA}) if \RuleSet is \MFA (resp. \EMFA).
\end{definition}

Even though our definition of \MFA may seem different from its original definition, these two  coincide.
Namely, a TGD set \RuleSet is \MFA with respect to Definition 3 in \cite{mfa} iff it is \MFA with respect to the above definition.
This equivalence readily follows from Proposition 5 in \cite{mfa}.


In the remainder of the section, we show some results about the complexity of checking \EMFA membership and of reasoning with \EMFA rule sets.
We apply the following preliminary lemma in the proofs of some of these results.

\begin{lemma}
\label{lemma:aux}
For a rule set \RuleSet, the number of atoms without cyclic terms that can be defined using  \FirstItem a finite set of constants $C$, \SecondItem the (Skolem) function symbols in $\EntitiesIn{\Functions}{\Sk{\RuleSet}}$, and \ThirdItem the predicates in \EntitiesIn{\Predicates}{\RuleSet} is doubly exponential in \RuleSet.
\end{lemma}
\begin{proof}
Firstly, we determine that the size of the set of all non-cyclic terms $T$ that can be defined using the symbols in $C$ and  \EntitiesIn{\Functions}{\Sk{\RuleSet}} is doubly exponential in \RuleSet.
Let $n = \vert \EntitiesIn{\Functions}{\Sk{\RuleSet}} \vert$ and let $m$ be the maximal arity of a function in \EntitiesIn{\Functions}{\Sk{\RuleSet}}.
By assumption $(\dagger)$ in the second section, we have that $\vert \EntitiesIn{\Functions}{\Sk{\RuleSet}} \vert \leq \vert \EntitiesIn{\Variables}{\RuleSet} \vert$.
For every $t \in \Terms$, let $r_t = t$ if $t \in \Constants$ and $r_t = f$ if $t$ is of the form $f(u_1, \ldots, u_n)$.
Moreover, let $G_t = (V_t, E_t)$ be a directed tree defined as follows: if $t \in \Constants$, then $G_t = (\{t\}, \emptyset)$; if $t$ is of the form $f(u_1, \ldots, u_n)$, then
$$G_t = (\{f\} \cup \bigcup_{1 \leq i \leq n} V_{u_i}, \{\Tuple{f, r_{u_i}} \mid 1 \leq i \leq n\} \cup \bigcup_{1 \leq i \leq n} E_{u_i}).$$
For every $t \in \Terms$, the maximal length of a path in $G_t$ is at most $n$, since a greater length implies the repetition of a function symbol, which in turn would imply that $t$ is cyclic.
Hence, $G_t$ is of depth at most $n$, fan-out at most $m$, and with leafs labelled with constants in $C$.
Such trees have at most $m^n$ leafs and at most $n \cdot m^n$ nodes in total.
As each node is assigned a constant or function symbol, there are at most $(\vert C \vert + n)^{(n \cdot m^n)}$ such trees and hence, non-cyclic terms, overall.
Therefore, $\vert T \vert$ is doubly exponential in \RuleSet.

Secondly, we determine the number of atoms without cyclic terms.
Let $k = \vert \EntitiesIn{\Predicates}{\RuleSet} \vert$ and let $\ell$ be the maximal arity of a predicate in \RuleSet.
Then, the number of atoms without cyclic terms is at most $k \vert T\vert^\ell = k(\vert C \vert + n)^{\ell \cdot n \cdot m^n}$.
\end{proof}

Despite the fact that \EMFA can be applied to rule sets with equality, checking membership with respect to this notion is not harder than deciding \MFA membership.

\begin{theorem}
\label{theorem:reasoning-membership}
Deciding whether a rule set \RuleSet is \EMFA is \DoubleExpTime-complete.
\end{theorem}
\begin{proof}
(Membership)
To decide if \RuleSet is \EMFA, it suffices to compute the atom set \EMFASet{\RuleSet} up until the occurrence of an atom containing a cyclic term.
By Lemma~\ref{lemma:aux}, we may have to compute \EMFASet{\RuleSet} up until it contains doubly exponential many atoms.
To decide whether to include each of these atoms we have to determine whether there is a pair \Tuple{\Rule, \Subs} with $\Rule \in \RuleSet$ that is applicable to some subset of \EMFASet{\RuleSet}.
Checking if this is the case can be done in \DoubleExpTime and hence, the \EMFA membership check can be performed in \DoubleExpTime.

(Hardness)
An equality-free rule set (i.e., a TGD set) is \EMFA iff it is \MFA.
Hence, an algorithm that solves \EMFA membership can be used to decide \MFA membership.
Therefore, the hardness of the \EMFA check follows from the hardness of the \MFA membership check, which was shown to be \DoubleExpTime-hard by Cuenca Grau et al. \shortcite{mfa} (see Theorem~8).
\end{proof}

In the proof of the following result, we show that the chase is a decision procedure for BCQ entailment over \EMFA ontologies that runs in \DoubleExpTime.

\begin{theorem}
\label{theorem:bcq-membership}
Deciding BCQ entailment over an \EMFA ontology is in \DoubleExpTime.
\end{theorem}
\begin{proof}
The above result follows from the fact that, for any (arbitrarily chosen) chase \Chase of an \EMFA ontology $\On = \Tuple{\RuleSet, \FactSet}$, the atom set \Chase does not contain any cyclic terms.
Hence, we can show that this atom set \Chase can be computed in \DoubleExpTime with an analogous argument to the one that is used in the ``Membership'' part of the proof of Theorem~\ref{theorem:reasoning-membership} to show that \EMFASet{\RuleSet} can be computed in \DoubleExpTime.
Note that, once an atom is removed from a chase sequence due to the application of a tuple with an EGD, it may never be reintroduced in any descendant in the sequence by Definition~\ref{definition:chase}.
We show that \Chase does not contain cyclic terms via indcution.

By Definition~\ref{definition:chase}, $\Chase$ is the union of all of the sets in some chase sequence $\AtomSet_0, \AtomSet_1, \ldots$ of $\On$.
The fact that \Chase does not contain cyclic terms follows from the following claim:
for all $i \geq 0$, $\AtomSet_i \Rewriting_\star \subseteq \EMFASet{\RuleSet}$ with $\Rewriting_\star$ the rewriting that maps every ground term $t$ to the term that results from replacing every syntactic occurrence of a constant with $\star$ (e.g., $A(a, f(a))\Rewriting_\star = A(\star, f(\star))$ where $a \in \Constants$).
Note that, if \RuleSet is \EMFA, then $\EMFASet{\RuleSet}$ does not contain cyclic terms.

(Base case)
By Definition~\ref{definition:chase}, $\AtomSet_0 = \FactSet$.
By Definition~\ref{definition:emfa}, $\CI{\RuleSet} \subseteq \EMFASet{\RuleSet}$ and hence, $\AtomSet_0\Rewriting_\star \subseteq \EMFASet{\RuleSet}$ since $\AtomSet_0\Rewriting_\star \subseteq \CI{\RuleSet}$.

(Inductive step)
Let $i \geq 1$.
Then, there is a rule $\Rule \in \RuleSet$ and a substitution \Subs such that \Tuple{\Rule, \Subs} is applicable to $\AtomSet_{i-1}$ and $\AtomSet_i$ is the application of \Tuple{\Rule, \Subs} on $\AtomSet_{i-1}$.
By induction hypothesis, we have that $\AtomSet_{i-1}\Rewriting_\star \subseteq \EMFASet{\RuleSet}$.
We consider two different cases depending on whether \Rule is a TGD or an EGD.
\begin{itemize}
\item Let \Rule be a TGD; that is, this rule is of the form $\Body \to \exists \Vw . \Head$.
Then, $\Body\Subs \subseteq \AtomSet_{i-1}$ and $\AtomSet_i = \AtomSet_{i-1} \cup \{\SkApp{\Subs}{\Head}\}$.
Moreover, $(\Body\Subs)\Rewriting_\star \subseteq \EMFASet{\RuleSet}$ since $\AtomSet_{i-1}\Rewriting_\star \subseteq \EMFASet{\RuleSet}$.
Hence, $(\SkApp{\Subs}{\Head}) \Rewriting_\star \subseteq \EMFASet{\RuleSet}$, and $\AtomSet_i\Rewriting_\star$ is a subset of $\EMFASet{\RuleSet}$.

\item Let \Rule be an EGD; that is, this rule is of the form $\Body \to x \approx y$.
Then, $\Body\Subs \subseteq \AtomSet_{i-1}$, and $\Body\Subs\Rewriting_\star \subseteq \EMFASet{\RuleSet}$ since $\AtomSet_{i-1}\Rewriting_\star \subseteq \EMFASet{\RuleSet}$.
We consider two different cases.
\begin{itemize}
\item $\Subs(x) \prec \Subs(y)$ with $\prec$ the strict total order introduced before Definition~\ref{definition:applicability}.
Then, $\AtomSet_i = \AtomSet_i[\Subs(y) / \Subs(x)]$.
Moreover, $\EMFASet{\RuleSet}[\Rewriting_\star(\Subs(y)) / \Rewriting_\star(\Subs(x))] \subseteq \EMFASet{\RuleSet}$ since $\Subs(x) \prec \Subs(y)$ implies $\Depth{\Subs(x)} \leq \Depth{\Subs(y)}$.
\item $\Subs(y) \prec \Subs(x)$.
Analogous to the previous case.
\end{itemize}
\end{itemize}
In either case, $\AtomSet_i\Rewriting_\star$ is a subset of $\EMFASet{\RuleSet}$.
\end{proof}

As implied by the  following result, the chase is a worst-case optimal procedure to reason with \EMFA ontologies.

\begin{theorem}
\label{theorem:bcq-hard}
Deciding BCQ entailment for ontologies \Tuple{\RuleSet, \FactSet} with \RuleSet an \EMFA rule set is \DoubleExpTime-hard.
\end{theorem}
\begin{proof}
Hardness is established by modifying the construction of a \DoubleExpTime Turing machine given for weakly acyclic rules by \cite{qa-non-guarded-rules}.
For a more detailed explanation of this argument, see the proof of Theorem~3 in \cite{rmfa}.
\end{proof}

In the remainder of the section, we present some results and examples that illustrate the expressivity of \EMFA compared to that of using \MFA over axiomatised rule sets.

\begin{theorem}
A rule set \RuleSet is \EMFA if \St{\RuleSet} is \MFA.
\end{theorem}
\begin{proof}
Let $\EMFASet{\RuleSet}^0 = \CI{\RuleSet}, \EMFASet{\RuleSet}^1, \ldots$ be a sequence consisting of all the intermediate sets that are computed to construct the set \EMFASet{\RuleSet} by applying the rules defined in Definition~\ref{definition:emfa}.
We show that $\EMFASet{\RuleSet}^i \subseteq \EMFASet{\St{\RuleSet}}$ for all $i \geq 1$ via induction.
Hence, $\EMFASet{\RuleSet} \subseteq \EMFASet{\St{\RuleSet}}$ and the theorem follows.

(Base Case)
By Definition~\ref{definition:emfa}, $\CI{\St{\RuleSet}} \subseteq \EMFASet{\St{\RuleSet}}$.
Therefore, we have that $\EMFASet{\RuleSet}^0 \subseteq \EMFASet{\St{\RuleSet}}$ since $\CI{\RuleSet} \subseteq \CI{\St{\RuleSet}}$.

(Inductive Step)
Let $i \geq 1$.
Then, by IH we have that $\EMFASet{\RuleSet}^{i-1} \subseteq \EMFASet{\St{\RuleSet}}$.
We consider the following cases.
\begin{itemize}
\item $\EMFASet{\RuleSet}^{i} = \SkApp{\Subs}{\Head} \cup \EMFASet{\RuleSet}^{i-1}$ where $\Head$ is the head of some TGD $\Rule = \Body \to \exists \Vw . \Head \in \RuleSet$ and $\Subs$ is some substitution such that $\Body\Subs \subseteq \EMFASet{\RuleSet}^{i-1}$.
By Definition~\ref{definition:standard}, $\Rule \in \St{\RuleSet}$.
Since $\Body\Subs \subseteq \EMFASet{\St{\RuleSet}}$ by IH, $\SkApp{\Subs}{\Head} \subseteq \EMFASet{\St{\RuleSet}}$.
\item $\EMFASet{\RuleSet}^i = \EMFASet{\RuleSet}^{i-1}[\Subs(x) / \Subs(y)] \cup \EMFASet{\RuleSet}^{i-1}$, there is an EGD $\Body \to x \approx y \in \RuleSet$ and a substitution $\Subs$ with $\Body\Subs \subseteq \EMFASet{\RuleSet}^{i-1}$, and $\Depth{\Subs(x)} \leq \Depth{\Subs(y)}$.
By Definition~\ref{definition:standard}, $\Body \to \EP(x, y) \in \St{\RuleSet}$.
Since $\Body\Subs \subseteq \EMFASet{\St{\RuleSet}}$ by IH, $\EP(x, y)\Subs \in \EMFASet{\St{\RuleSet}}$.
Since TGD $\eqref{rule:symmetry} \in \St{\RuleSet}$, $\EP(\Subs(y), \Subs(x)) \in \EMFASet{\St{\RuleSet}}$.
Because of the rules of type \eqref{rule:replacement} in \St{\RuleSet}, $\varphi[\Subs(x) / \Subs(y)] \in \EMFASet{\St{\RuleSet}}$ for all $\varphi \in \EMFASet{\St{\RuleSet}}$.
\item $\EMFASet{\RuleSet}^i = \EMFASet{\RuleSet}^{i-1}[\Subs(y) / \Subs(x)] \cup \EMFASet{\RuleSet}^{i-1}$, there is an EGD $\Body \to x \approx y \in \RuleSet$ and a substitution $\Subs$ with $\Body\Subs \subseteq \EMFASet{\RuleSet}^{i-1}$, and $\Depth{\Subs(y)} \leq \Depth{\Subs(x)}$.
Analogous to the previous case.
\end{itemize}
In either case, $\EMFASet{\RuleSet}^i \subseteq \EMFASet{\St{\RuleSet}}$.
\end{proof}

As shown by the following example, the converse of the above theorem does not hold.

\begin{example}
The rule set \RuleSet from Example~\ref{example:1} is \EMFA, but the TGD set $\St{\RuleSet}$ is not \MFA.
\end{example}

To conclude the section, we introduce some examples that illustrate the generality of \EMFA versus that of applying \MFA over singularised rule sets.
For instance, there are rule sets that are \EMFA, but no singularisation of these are \MFA.

\begin{example}
The rule set \RuleSet containing all of the following rules is \EMFA, but no TGD set in \Sing{\RuleSet} is \MFA.
\begin{align}
A(x) &\to \exists v . R(x, v) \wedge B(v) \\
A(x) &\to \exists w . S(x, w) \wedge C(w) \\
C(x) \wedge B(x) &\to A(x) \\
R(x, y) &\to x \approx y \\
S(x, y) &\to x \approx y
\end{align}
For example, the chase of an ontology $\Tuple{\RuleSet, \FactSet}$ with $\RuleSet \in \Sing{\RuleSet}$ and $\FactSet = \{A(a), R(a, a), S(a, a)\}$ does not terminate irrespectively of \RuleSet.
Note that, this is the case even though all such ontologies admit finite chases.
Therefore, neither TGD set in \Sing{\RuleSet} is \MFA.
\end{example}

Furthermore, there are rule sets that are not \EMFA, but all of their singularisations are \MFA.

\begin{example}
\label{example:last}
Even though the rule set \RuleSet with all of the following rules is not \EMFA, all of the sets in $\Sing{\RuleSet}$ are \MFA.
\begin{align}
A(x) &\to \exists v . R(x, v) \wedge B(v) \\
B(x) &\to \exists w . R(x, w) \wedge C(w) \\
R(x, y) \wedge R(x, z) &\to y \approx z
\end{align}
Note that, the two TGD sets in \Sing{\RuleSet} are equivalent.
\end{example}


\section{Evaluation}
\label{section:evaluation}

\begin{figure*}
    \small
\begin{center}
\small{\begin{tabular}{ c |@{\hspace{0.25\tabcolsep}} c |@{\hspace{0.25\tabcolsep}} c |@{\hspace{0.25\tabcolsep}} c}
\backslashbox{{\small \#TGDs}}{{\small \#EGDs}}
						& $[1, 2]$			& [3, 7] 			& $[\geq 8]$ \\\hline
\multirow{2}{*}{$[1, 2]$}		& 106 / 101 / 101	& 56 / 50 / 50		& 8 / 7 / 7 \\
						& 6ms / 9ms		& 13ms / 21ms		& 28ms / 46ms \\ \hline
\multirow{2}{*}{$[3, 15]$}		& 89 / 65 / 65		& 92 / 74 / 74		& 87 / 66 / 66 \\
						& 25ms / 51ms		& 31ms / 71ms		& 41ms / 65ms \\ \hline
\multirow{2}{*}{$[\geq 16]$}	& 13 / 7 / 7		& 12 / 3 / 3		& \textbf{102 / 14 / 11} \\
						& 404ms / 606ms	& 56ms / 94ms		& 2.7s / 7.6s
\end{tabular}}
\quad
\small{\begin{tabular}{ c |@{\hspace{0.25\tabcolsep}} c |@{\hspace{0.25\tabcolsep}} c |@{\hspace{0.25\tabcolsep}} c}
\backslashbox{{\small \#TGDs}}{{\small \#EGDs}}
					& $[1]$			& [2, 3] 			& $[\geq 4]$ \\\hline
\multirow{2}{*}{$[1]$}		& 5 / 5 / 5			& 1 / 1 / 1			& 1 / 1 / 1 \\
					& 2ms / 4ms		& 5ms / 11ms		& 6ms / 10ms \\ \hline
\multirow{2}{*}{$[2, 5]$}	& 6 / 5 / 5			& 9 / 8 / 8			& 5 / 5 / 5 \\
					& 10ms / 18ms		& 5ms / 9ms		& 10ms / 15ms \\ \hline
\multirow{2}{*}{$[\geq 6]$}	& 82 / 9 / 9		& 2 / 2 / 2			& 20 / 2 / 2 \\
					& 62s / 78s		& 16ms / 27ms		& \textbf{15s / 124s}
\end{tabular}}
\end{center}
\caption{MOWLCorp (left) and Oxford Ontology Repository (right) results, see Summary~\ref{results} for an explanation of the above}
\label{table:experiments}
\end{figure*}

We performed a number of experiments to verify, from an empirical perspective, claims (\ClaimSing--\ClaimEMFAPerformance) stated in the introduction.
All the used rule sets are available online.\footnote{\url{github.com/dcarralma/existential-rule-sets-with-equality}}

To verify (\ClaimSing), we implemented the ``renaming'' chase variant presented in Definition~\ref{definition:chase} in VLog \cite{vlog}, which is an efficient rule engine for existential rules \cite{vlog-chase}.
Then, we checked if using this procedure to compute the chase over an ontology $\Tuple{\RuleSet, \FactSet}$ is more efficient than computing the chase of an ontology $\Tuple{\RuleSet', \FactSet}$ with $\RuleSet'$ some arbitrarily chosen rule set in \Sing{\RuleSet}.
For this experiment we considered two ontologies---DBPedia \cite{dbpedia} and Claros \cite{claros}---that we obtained from the evaluation in \cite{rewriting}.
In either case, the performance of the ``renaming'' chase was far superior: we can compute the chase of DBPedia in 27.5s when using renaming to deal with equality; computing the chase of a (randomly selected) singularisation of this ontology takes 590s.
We get similarly lopsided results for Claros: 11.7s when using renaming; 67.4s with singularisation.

\citeauthor{rewriting} \shortcite{rewriting} showed that using renaming to deal with equality is more efficient than applying the standard axiomatisation.
Therefore, we conclude that the use of axiomatisations (singularisation and standard) results in poor performance; a fact that validates the practical usefulness of Theorems~\ref{theorem:standard} and \ref{theorem:singularisation}.
As discussed in the introduction, the application of these results allows us to check acyclicity with respect to some axiomatisation of a rule set, and then use the original rule set for computing the chase.

To verify claims (\ClaimStandard--\ClaimEMFAPerformance), we use Description Logics TBoxes \cite{dl-handbook} from MOWLCorp \cite{mowl-corpus} and the Oxford Ontology Repository\footnote{\url{www.cs.ox.ac.uk/isg/ontologies/}} (OOR).
First, we normalise these TBoxes by structural decomposition of complex axioms into the normal form considered in \cite{elifying} and subsequently filter all TBoxes with non-deterministic features or nominals.
Then, we apply a standard translation into FOL to obtain equivalent existential rule sets (see Section 6 in \cite{mfa}).
Finally, we discard rule sets that do not contain at least one EGD and at least one TGD with existential quantified variables.

We implemented the \EMFA and \MFA checks in VLog and applied them to the rule sets from MOWLCorp and OOR.
In summary, we obtained the following results:
\begin{itemize}
\item MOWLCorp: out of a total of 565 rule sets after preprocessing, we found that 387 rule sets are \EMFA.
Moreover, 73 and 384 rule sets are \MFA if we apply the standard axiomatisation and singularisation, respectively.
\item OOR: out of 131 rule sets, we found that that 38 rule sets are \EMFA.
Also, 38 are \MFA if singularisation is applied.
\end{itemize}
For each rule set \RuleSet, we compute a single (arbitrarily chosen) set in \Sing{\RuleSet} when testing singularisation.
Given the poor performance of the \MFA check with the standard axiomatisation for MOWLCorp, we did not consider this technique again when we repeated the experiment with OOR.

From the above results, we can readily verify claim (\ClaimStandard).
All rule sets classified as \MFA when applying either axiomatisation technique were also found to be \EMFA and therefore, we consider that claim (\ClaimEMFAGenerality) was also validated.

To verify (\ClaimEMFAPerformance), we measured the time that took to perform each check.
On average, the \EMFA check takes 35\% and 60\% of the time taken by the ``\MFA{} + singularisation'' check for the rule sets in MOWLCorp and OOR, respectively.
We present a more detailed analysis in Figure~\ref{table:experiments}.
\begin{results}
\label{results}
In each cell in Figure~\ref{table:experiments}, we include information about some subset of the rule sets in MOWLCorp (left table) or OOR (right table).
For example, the upper right cell in the left table contains the counts and average times for the rule sets in MOWLCorp that contain at least 8 EGDs and between $1$ and $2$ TGDs with existential quantifiers.
Each cell contains two lines: the first features the total count of rule sets included, as well as the number of \EMFA and ``\MFA{} + singularisation'' successful checks; the second one includes two values which indicate the average times taken by the \EMFA and the ``\MFA{} + singularisation'' tests.
\end{results}

Note that, \EMFA outperforms ``\MFA{} + singularisation''  by almost an order of magnitude on average for the hardest rule sets considered (i.e., lower right cell in the right table). Moreover, in a small number of cases, \EMFA succeeded when \MFA failed (lower right cell in the left table).

\section{Related Work and Conclusions}
\label{section:conclusions}

A previously existing acyclicity notion that can be directly applied to rule sets with EGDs is \emph{weak acyclicity} \cite{wa}.
Alas, this notion is significantly less general than checking \MFA membership over singularised rule sets in practice (see Section~7 of \cite{mfa}).

As for future work, we plan to extend \emph{restricted model-faithful acyclicity} (\RMFA) \cite{rmfa}, an acyclicity notion for the \emph{Datalog-first restricted chase}, so it can be applied to rule sets with equality.
Since \RMFA is more general than \MFA, this extension can yield an even more general condition applicable for rule sets with EGDs.
To verify that this notion captures all possible rule sets with a terminating chase, we plan to develop a \emph{cyclicity notion} such as the one presented in \cite{rmfa}.
That is, a sufficient condition that can detect if the chase does not terminate for a given rule set.

In this paper, we have presented several results that we believe are of theoretical interest and of practical usefulness regarding chase termination of rule sets with EGDs.
In particular, we believe that Theorems~\ref{theorem:standard} and \ref{theorem:singularisation} are very useful, as they enable the application of all existing acyclicity notions to existential rule sets with equality.



\paragraph*{Acknowledgments}
This work is funded by Deutsche Forschungsgemeinschaft (DFG) grant 389792660 as part of TRR~248 (see \url{www.perspicuous-computing.science}) and by the NWO research programme 400.17.605 (VWData).
We also thank Irina Dragoste for her useful comments.

\bibliographystyle{style/aaai}
\bibliography{bibliography/reference}

\begin{tr}
\newpage

\onecolumn

\begin{list}{}{
\setlength{\topsep}{0pt}
\setlength{\leftmargin}{1cm}
\setlength{\rightmargin}{1cm}
\setlength{\listparindent}{\parindent}
\setlength{\itemindent}{\parindent}
\setlength{\parsep}{\parskip}}
\item[]

\makeatother
\appendix
\large

\section{Formal Proofs}

In this section we include the proofs for Theorems~\ref{theorem:standard} and \ref{theorem:singularisation}.
Moreover, we also include and proof the two following auxiliary lemmas, which are later used in the proof of Theorem~\ref{theorem:standard}.

\begin{lemma}
\label{lemma:aux}
Consider an \EPComplete atom set \AtomSet and some terms $t, u \in \EntitiesIn{\Terms}{\AtomSet}$.
If the atom $\EP(t, u)$ is in \AtomSet, then $\Rewriting_\AtomSet(t) = \Rewriting_\AtomSet(u)$.
\end{lemma}
\begin{proof}
Proof by contradiction.
\begin{enumerate}
\item Let $\AtomSet$ be an \EPComplete atom set.
\item Assume that there are some terms $t, u \in \EntitiesIn{\Terms}{\AtomSet}$ with (a) $\EP(t, u) \in \AtomSet$ and (b) $\Rewriting_\AtomSet(t) \neq \Rewriting_\AtomSet(u)$.
\item By (2) and Definition~\ref{definition:rewriting-standard}: $\EP(t, \Rewriting_\AtomSet(t)), \EP(u, \Rewriting_\AtomSet(u)) \in \AtomSet$.
\item By (1-3) and Definition~\ref{definition:epcompleteness}: $\EP(t, \Rewriting_\AtomSet(u)), \EP(u, \Rewriting_\AtomSet(t)) \in \AtomSet$.
\item By (2.b) and (4): $\Rewriting_\AtomSet(u) \prec \Rewriting_\AtomSet(t)$ and $\Rewriting_\AtomSet(t) \prec \Rewriting_\AtomSet(u)$ with $\prec$ the strict total order over the set of terms introduced before Definition~\ref{definition:applicability}.
\item (5) results in a contradiction and hence, the assumption from (2) does not hold.
\end{enumerate}
\end{proof}

\begin{lemma}
\label{lemma:aux-2}
Consider some \EPComplete atom set \AtomSet and some term $t \in \EntitiesIn{\Terms}{\AtomSet}$.
If $t$ is in the range of $\Rewriting_\AtomSet$, then $\Rewriting_\AtomSet(t) = t$.
\end{lemma}
\begin{proof}
Proof by contradiction.
\begin{enumerate}
\item Let \AtomSet be an \EPComplete atom set.
\item Assume that there are some $t, u \in \EntitiesIn{\Terms}{\AtomSet}$ with (a) $t$ in the range of $\Rewriting_\AtomSet$, (b) $\Rewriting_\AtomSet(t) = u$, and (c) $t \neq u$.
\item By (2.a): there is some $v \in \EntitiesIn{\Terms}{\AtomSet}$ with $\Rewriting_\AtomSet(v) = t$.
\item By (2.b), (3), and Definition~\ref{definition:rewriting-standard}: $\EP(t, u), \EP(v, t) \in \AtomSet$.
\item By (1), (4), and Definition~\ref{definition:epcompleteness}: $\EP(t, t), \EP(v, u) \in \AtomSet$.
\item By (2.b), (2.c), (5), and Definition~\ref{definition:rewriting-standard}: $u \prec t$ with $\prec$ the strict total order over the set of terms introduced before Definition~\ref{definition:applicability}.
\item By (2.c), (3), (5), and Definition~\ref{definition:rewriting-standard}: $t \prec u$.
\item (6) and (7) result in a contradiction and hence, the assumption from (2) does not hold.
\end{enumerate}
\end{proof}

\theoremstandard*
\begin{proof}
Set-up for the proof.
\begin{enumerate}
\item Premise: the chase of \St{\RuleSet} terminates.
\item Let \FactSet be some fact set and let $\On = \Tuple{\RuleSet, \FactSet}$.
Then, $\St{\On} = \Tuple{\St{\RuleSet}, \FactSet}$ by Definition~\ref{definition:standard}.
\item By (1) and (2): the chase of $\St{\On}$ terminates.
\item By (3): the number of chase steps for $\St{\On}$ is finite.
\item Assume that, for any chase step \AtomSet of \On, there is an \EPComplete chase step $\AtomSetAux$ of \St{\On} with $[\AtomSetAux] = \AtomSet$.
\item By (4) and (5): the chase of $\On$ terminates.
\item By (2) and (6): the chase of $\RuleSet$ terminates.
\item We show that the assumption in (3) holds with the following inductive argument.
\end{enumerate}
\smallskip

\noindent Base case:
\begin{enumerate}
\item Let \AtomSet be the first element in some chase sequence of \On.
\item By (1) and Definition~\ref{definition:chase}: $\AtomSet = \FactSet$.
\item By (2) and Definition~\ref{definition:chase}: \AtomSet is the first element of every chase sequence of \St{\On} and hence, \AtomSet is a chase step of \St{\On}.
\item Let $\AtomSetAux = \AtomSet \cup \{\EP(t, t) \mid t \in \EntitiesIn{\Terms}{\AtomSet}\}$.
\item By (3) and (4): $\AtomSetAux$ is a chase step for \St{\On}, as it can be obtained by exhaustively applying the rules of type \eqref{rule:reflexivity} in \St{\RuleSet} to \AtomSetAux.
Note that, $\EntitiesIn{\Predicates}{\FactSet} \subseteq \EntitiesIn{\Predicates}{\RuleSet}$ by the definition of an ontology.
Hence, for every $P \in \EntitiesIn{\Predicates}{\FactSet}$ there is a rule of type \eqref{rule:reflexivity} in $\St{\RuleSet}$ instantiated for the predicate $P$.
\item By Definition~\ref{definition:standard}: the set \FactSet does not contain facts over \EP.
\item By (2), (4), (6), and Definition~\ref{definition:epcompleteness}: the atom set \AtomSetAux is \EPComplete.
\item By (2), (4), (6), and Definition~\ref{definition:rewriting-standard}: $\Rewriting_\AtomSetAux$ is identity function over the set \EntitiesIn{\Terms}{\AtomSetAux}.
Thefore, $[\AtomSetAux] =  \AtomSet$.
\end{enumerate}
\smallskip

\noindent Inductive step:
\begin{enumerate}[1.]
\item Let \AtomSet be some chase step of \On that is not the first element in any chase sequence of \On.
\item By (1): there is some chase step $\AtomSet'$ that is the predecessor of \AtomSet in some chase sequence of \On.
\item By (2) and induction hypothesis: there is an \EPComplete chase step $\AtomSetAux''$ of \St{\On} with $[\AtomSetAux''] = \AtomSet'$.
\item Let $\AtomSetAux_0, \ldots, \AtomSetAux_n$ be some sequence of atom sets that results from exhaustively applying the rules of type (\ref{rule:replacement}) in \St{\RuleSet} to $\AtomSetAux''$.
That is, $\AtomSetAux_0, \ldots, \AtomSetAux_n$ is a sequence such that
\begin{enumerate}[a.]
\item $\AtomSetAux_0 = \AtomSetAux''$;
\item for all $i \in \{1, \ldots, n\}$, there is some rule $\Rule_i \in \St{\RuleSet}$ of type (\ref{rule:replacement}) and some substitution $\Subs_i$ such that $\AtomSetAux_i$ is the application of $\Tuple{\Rule_i, \Subs_i}$ on $\AtomSetAux_{i-1}$; and
\item $\AtomSetAux_n$ satisfies all of the rules of type (\ref{rule:replacement}) in \St{\RuleSet}.
\end{enumerate}
\item By Definition~\ref{definition:standard}: rules of type (\ref{rule:replacement}) do not contain existentially quantified variables.
Moreover, these do not contain the predicate \EP in the head.
\item By (4.b) and (5): for all $i \in \{1, \ldots, n\}$, if $\AtomSetAux_{i-1}$ is \EPComplete, then so is $\AtomSetAux_i$.
\item By (3), (4.a), and (6): $\AtomSetAux_i$ is \EPComplete for all $i \in \{0, \ldots, n\}$. 
\item We show via induction that $[\AtomSetAux_i] = \AtomSet'$ for all $i \in \{0, \ldots, n\}$.
\begin{itemize}
\item Base case: $[\AtomSetAux_0] = \AtomSet'$ by (3) and (4.a).
\item Inductive step:
\begin{enumerate}[a.]
\item Let $i \in \{1, \ldots, n\}$.
\item By (7): $\AtomSetAux_{i-1}$ and $\AtomSetAux_i$ are \EPComplete.
Hence, the atom sets $[\AtomSetAux_{i-1}]$ and $[\AtomSetAux_i]$ are well defined and can be used across the following argument.
\item By induction hypothesis: $[\AtomSetAux_{i-1}] = \AtomSet'$.
\item By (4.b): $\Rule_i$ is of the form $P(x_1, \ldots, x_j, \ldots, x_m) \wedge \EP(x_j, x_j') \to P(x_1, \ldots, x_j', \ldots, x_m)$ with $P \neq \mathop{\EP}$ an $m$-ary predicate and $j \in \{1, \ldots, m\}$.
\item By (4.b) and (d): $P(x_1, \ShortDots, x_m)\Subs_i, \EP(x_j, x_j')\Subs_i \in \AtomSetAux_{i-1}$.
\item By (5): $\Rewriting_{\AtomSetAux_i} = \Rewriting_{\AtomSetAux_{i-1}}$.
\item By (e), (f), and Lemma~\ref{lemma:aux}: $\Rewriting_{\AtomSetAux_i}(\Subs_i(x_j)) = \Rewriting_{\AtomSetAux_i}(\Subs_i(x_j'))$.
\item By (e), (f), and (g): $((P(x_1, \ldots, x_j', \ldots, x_m)\Subs_i)\Rewriting_{\AtomSetAux_i}  \in [\AtomSetAux_{i-1}]$.
\item By (4.b) and (d): $\AtomSetAux_i = \AtomSetAux_{i-1} \cup \{P(x_1, \ldots, x_j', \ldots, x_m)\Subs_i\}$
\item By (f), (h), (i), and Definition~\ref{definition:rewriting-standard}: $[\AtomSetAux_i] = [\AtomSetAux_{i-1}] \cup \{((P(x_1, \ldots, x_j', \ldots, x_m)\Subs_i)\Rewriting_{\AtomSetAux_i}\} = [\AtomSetAux_{i-1}]$.
\item By (c) and (j): $[\AtomSetAux_i] = \AtomSet'$.
\end{enumerate}
\end{itemize}
\item Let $\AtomSetAux' = \AtomSetAux_n$.
Note that, 
\begin{enumerate}[a.]
\item $\AtomSetAux'$ is \EPComplete by (7),
\item $[\AtomSetAux'] = \AtomSet'$ by (8), and
\item $\AtomSetAux'$ satisfies all of the rules of type \eqref{rule:replacement} in \RuleSet by (4.c).
\end{enumerate}
\item We show that $\AtomSet' \subseteq \AtomSetAux'$ by contradiction.
\begin{enumerate}[a.]
\item Assume that there is some fact $P(t_1, \ldots, t_m) \in \AtomSet'$ with $P(t_1, \ldots, t_m) \notin \AtomSetAux'$.
\item By (a) and (9.b): there is a fact of the form $P(u_1, \ldots, u_m) \in \AtomSetAux'$ with $[u_1 / t_1, \ldots, u_n / t_m] \subseteq \Rewriting_{\AtomSetAux'}$.
\item By (b) and Definition~\ref{definition:rewriting-standard}: $\EP(u_i, t_i) \in \AtomSetAux'$ for all $i \in \{1, \ldots, m\}$.
\item By (9.c): $\AtomSetAux'$ satisfies $P(x_1, \ShortDots, x_m) \wedge \EP(x_j, x_j') \to P(x_1, \ldots, x_j', \ldots, x_m)$ for all $j \in \{1, \ldots, m\}$.
\item By (b), (c), and (d): $P(t_1, \ldots, t_m) \in \AtomSetAux'$.
\item (a) and (e) result in a contradiction and hence, the assumption from (a) does not hold.
\end{enumerate}
\item By (2): there is a rule $\Rule \in \RuleSet$ and a substitution \Subs such that
\begin{enumerate}[a.]
\item \Subs is defined for all of the universally and none of the existentially quantified variables in \Rule,
\item $\Tuple{\Rule, \Subs}$ is applicable to $\AtomSet'$,
\item $\AtomSet = \Application{\AtomSet'}{\Rule}{\Subs}$, and
\item $\Body\Subs \subseteq \AtomSet'$ with \Body the body of \Rule.
\end{enumerate}
\item By (10) and (11.d): $\Body\Subs \subseteq \AtomSetAux'$.
\item We consider two cases (T) and (E), depending on whether the rule \Rule is a TGD or an EGD.
\begin{itemize}
\item[T.] Assume that \Rule is a TGD.
That is, \Rule is a rule of the form $\Body \to \exists \Vw . \Head \in \RuleSet$.
\begin{enumerate}[a.]
\item By (11) and Definition~\ref{definition:standard}: $\Rule \in \St{\RuleSet}$.
\item We show that \Tuple{\Rule, \Subs} is applicable to $\AtomSetAux'$ by contradiction.
\begin{enumerate}[I.]
\item Assume that $\Tuple{\Rule, \Subs}$ is not applicable to $\AtomSetAux'$.
\item By (I), (12), and Definition~\ref{definition:applicability}: there is a substitution $\SubsAux \supseteq \Subs$ with $\Head\SubsAux \subseteq \AtomSetAux'$.
\item By (II) and Lemma~\ref{lemma:aux-2}: $\Subs \subseteq (\Rewriting_{\AtomSetAux'} \circ \SubsAux)$.\footnote{The expression $\Rewriting_{\AtomSetAux'} \circ \SubsAux$ refers to the function such that $(\Rewriting_{\AtomSetAux'} \circ \SubsAux)(t) = \Rewriting_{\AtomSetAux'} (\SubsAux(t))$ for all $t$ in the domain of $\SubsAux$.}
Note that all of the terms in the range of \Subs are in $\AtomSet'$ and hence, these are also in the range of $\Rewriting_{\AtomSetAux'}$ by (9.b).
\item By (9.b) and (II): $\Head(\Rewriting_{\AtomSetAux'} \circ \SubsAux) \subseteq \AtomSet'$.
\item By (11.a) and (11.b): $\Head\SubsAux \not\subseteq \AtomSet'$ for all $\SubsAux \supseteq \Subs$.
\item (III), (IV), and (V) result in a contradiction and hence, the assumption from (I) does not hold.
\end{enumerate}
\item Let $\AtomSetAux = \Application{\AtomSetAux'}{\Rule}{\Subs} \cup \{\EP(t, t) \mid t \in \EntitiesIn{\Terms}{\Application{\AtomSetAux'}{\Rule}{\Subs}}\}$.
\item By (a), (b), and (c): \AtomSetAux is a chase step of \St{\On} as it can be obtained by exhaustively applying the rules of type \eqref{rule:reflexivity} in \St{\RuleSet} to the chase step \Application{\AtomSetAux'}{\Rule}{\Subs}.
\item By (T): the conjunction \Head does not contain the predicate \EP.
\item By (9.a), (c), (e), and Definition~\ref{definition:epcompleteness}: \AtomSetAux is \EPComplete.
\item By (e): $\Rewriting_\AtomSetAux = \Rewriting_{\AtomSetAux'} \cup \{\EP(t, t) \mid t \in \EntitiesIn{\Terms}{\AtomSetAux \setminus \AtomSetAux'}\}$.
\item By (9.b), (11.c), (c), (g), and Definition~\ref{definition:rewriting-standard}: $[\AtomSetAux] = \AtomSet$.
\end{enumerate}
\item[E.] Assume that \Rule is an EGD.
That is, \Rule is a rule of the of the form $\Body \to x \approx y$.
\begin{enumerate}[a.]
\item By (11) and Definition~\ref{definition:standard}: $\Rule' = \Body \to \EP(x, y) \in \St{\RuleSet}$.
\item We show that \Tuple{\Rule', \Subs} is applicable to $\AtomSetAux'$ by contradiction.
\begin{enumerate}[I.]
\item Assume that \Tuple{\Rule', \Subs} is not applicable to $\AtomSetAux'$.
\item By (12) and (I): $\EP(x, y)\Subs \in \AtomSetAux'$.
\item By (II) and Definition~\ref{definition:rewriting-standard}: $\Rewriting_{\AtomSetAux'}(\Subs(x)) = \Rewriting_{\AtomSetAux'}(\Subs(y))$.
\item By (11.b): $\Subs(x) \neq \Subs(y)$, and both $\Subs(x)$ and $\Subs(y)$ occur in $\AtomSet'$.
\item By (9.b), (IV), and Lemma~\ref{lemma:aux-2}: $\Rewriting_{\AtomSetAux'}(\Subs(x)) = \Subs(x)$ and $\Rewriting_{\AtomSetAux'}(\Subs(y)) = \Subs(y)$.
\item By (III) and (V): $\Subs(x) = \Subs(y)$.
\item (IV) and (VI) result in a contradiction and hence, the assumption from (I) does not hold.
\end{enumerate}
\item By (a) and (b): $\Application{\AtomSetAux'}{\Rule'}{\Subs} = \AtomSetAux' \cup \{\EP(\Subs(x), \Subs(y))\}$ is a chase step of \St{\On}.
\item Let $\AtomSetAux$ be the chase step of $\St{\On}$ that results from exhaustively applying the rules of type (\ref{rule:reflexivity}-\ref{rule:transitivity}) in \St{\RuleSet} to $\Application{\AtomSetAux'}{\Rule'}{\Subs}$.
\item By (d) and Definition~\ref{definition:epcompleteness}: the atom set $\AtomSetAux$ is \EPComplete.
\item By (9.b), (11.c), (d) and Definition~\ref{definition:rewriting-standard}: $[\AtomSetAux] = \AtomSet$.
\end{enumerate}
\end{itemize}
\end{enumerate}
\end{proof}

\theoremsingularisation*
\begin{proof}
Set-up for the proof.
\begin{enumerate}
\item Assume that there is some rule set $\RuleSet' \in \Sing{\RuleSet}$ for which the chase terminates.
\item Let \FactSet be some fact set, let $\On = \Tuple{\RuleSet, \FactSet}$, let $\On' = \Tuple{\RuleSet', \FactSet}$, and let $\AtomSet_0, \AtomSet_1, \ldots$ be some chase sequence of \On.
\item By (1) and (2): the chase of $\On'$ terminates.
\item Assume that there is a sequence of atom sets $\AtomSetAux_0, \AtomSetAux_1, \ldots$ and a sequence of rewritings $\Rewriting_0, \Rewriting_1, \ldots$ such that, for all $i \in \{0, \ldots, n\}$, 
\begin{enumerate}[a.]
\item $\AtomSetAux_i$ is a chase step of $\On'$,
\item $\AtomSetAux_i\Rewriting_i \setminus \{\EP(t, u) \mid t, u \in \Terms\} = \AtomSet_i$, and
\item $\EP(t, u) \in \AtomSetAux_i$ for all $t, u \in \Terms$ with $\Rewriting_i(t) = \Rewriting_i(u)$.
\end{enumerate}
Note that $\AtomSetAux_0, \AtomSetAux_1, \ldots$ may not be a chase sequence for $\On'$.
\item By (3): the number of chase steps of $\On'$ is finite.
\item By (3.a) and (5): the sequence $\AtomSetAux_0, \AtomSetAux_1, \ldots$ is finite.
\item By (2), (4.b), and (6): the sequence $\AtomSet_0, \AtomSet_1, \ldots$ is finite.
\item By (2) and (7): the chase of $\On$ terminates.
\item By (2) and (8): the chase of $\RuleSet$ terminates.
\item We show that the assumption in (3) holds with the following inductive argument.
\end{enumerate}
\smallskip

\noindent Base case:
\begin{enumerate}
\item Let $\AtomSetAux_0 = \FactSet \cup \{\EP(t, t) \mid t \in \EntitiesIn{\Terms}{\FactSet}\}$.
\item Let $\Rewriting_0$ be the identity function over $\EntitiesIn{\Terms}{\AtomSetAux_0}$.
\item By Definition~\ref{definition:chase}: \FactSet is a chase step of $\On'$.
\item By Definition~\ref{definition:chase}: $\FactSet = \AtomSet_0$.
\item By the definition of an ontology: $\EntitiesIn{\Predicates}{\FactSet} \subseteq \EntitiesIn{\Predicates}{\RuleSet}$.
\item By (5) and Definition~\ref{definition:singularisation}: for every $P \in \EntitiesIn{\Predicates}{\FactSet}$, there is a rule of type \eqref{rule:reflexivity} in $\RuleSet'$.
\item By (1), (3), and (6): $\AtomSetAux_0$ is a chase step of $\On'$ as it can be obtained by exhaustively applying the rules of type \eqref{rule:reflexivity} in $\RuleSet'$ to \FactSet.
\item By (1), (2), and (4): $\AtomSetAux_0\Rewriting_0 \setminus \{\EP(t, u) \mid t, u \in \Terms\} = \AtomSet_0$.
\item By (1) and (2): $\EP(t, u) \in \AtomSetAux_0$ for all $t, u \in \Terms$ in $\AtomSetAux_0$ with $\Rewriting_0(t) = \Rewriting_0(u)$.
\end{enumerate}
\smallskip

\noindent Inductive step:
\begin{enumerate}
\item Let $i \geq 1$.
\item By induction hypothesis: there is an atom set $\AtomSetAux_{i-1}$ and a rewriting $\Rewriting_{i-1}$ such that
\begin{enumerate}[a.]
\item $\AtomSetAux_{i-1}$ is a chase step of $\On'$,
\item $\AtomSetAux_{i-1}\Rewriting_{i-1} \setminus \{\EP(t, u) \mid t, u \in \Terms\} = \AtomSet_{i-1}$, and
\item $\EP(t, u) \in \AtomSetAux_{i-1}$ for all $t, u \in \Terms$ with $\Rewriting_{i-1}(t) = \Rewriting_{i-1}(u)$.
\end{enumerate}
\item Since $\AtomSet_0, \AtomSet_1, \ldots$ is a chase sequence of \On, there is some rule $\Rule = \Body[\Vx] \to H \in \RuleSet$ with $\Vx = x^1, \ldots, x^n$ and some substitution \Subs such that
\begin{enumerate}[a.]
\item $\Tuple{\Rule, \Subs}$ is applicable to $\AtomSet_{i-1}$, and
\item $\AtomSet_i$ is the application of $\Tuple{\Rule, \Subs}$ on $\AtomSet_{i-1}$.
\end{enumerate}
\item By (3.a): $\Body\Subs \subseteq \AtomSet_{i-1}$.
\item By (3) and Definition~\ref{definition:singularisation}: since $\RuleSet' \in \Sing{\RuleSet}$, there is some rule $\Body'[x^1, \Vx^1, \ldots, x^n, \Vx^n] \to H \in \RuleSet'$ such that $\Body' \in \Sing{\Body}$ and, for all $j \in \{1, \ldots, n\}$, the list $\Vx^j$ contains all of the variables of the form $x^j_k \in \EntitiesIn{\Variables}{\Body'}$ with $k \geq 1$.
\item By (5) and Definition~\ref{definition:singularisation}: every $x \in \EntitiesIn{\Variables}{\Body'}$ occurs in one atom in $\Body'$ defined over a predicate $P \neq \EP$.\footnote{These are the variables that are added when we compute the singularisation of a rule.}
\item By (2.b), (4), (5), and (6): there is some substitution $\Subs'$ such that
\begin{enumerate}[a.]
\item $\Rewriting_{i-1}(\Subs'(y)) = \Subs(x^j)$ for all $j \in \{1, \ldots, n\}$ and $y \in \Vx^j$,
\item $\Rewriting_{i-1}(\Subs'(x^j)) = \Subs(x^j)$ for all $j \in \{1, \ldots, n\}$, and
\item $\Fact\Subs' \subseteq \AtomSetAux_{i-1}$ for all $\Fact \in \Body'$ that are not defined over the predicate $\EP$.
\end{enumerate}
\item By (2.c), (7.a), and (7.b): $\EP(\Subs'(x^j), \Subs'(y)) \in \AtomSetAux_{i-1}$ for all $j \in \{1, \ldots, n\}$ and $y \in \Vx^j$.
\item By (7.c), (8), and Definition~\ref{definition:singularisation}: $\Body'\Subs' \subseteq \AtomSetAux_{i-1}$.
\item We consider two different cases (T) and (E), depending on whether $\Rule$ is a TGD or an EGD.
\begin{itemize}
\item[T.] Let $\Rule$ be a TGD.
That is, $\Rule$ is a formula of the form $\Body[\Vx] \to \exists \Vw . \Head[\Vy, \Vw]$ with $\Vy \subseteq \Vx$.
\begin{enumerate}[a.]
\item We show that $\Tuple{\Rule', \Subs'}$ is applicable to $\AtomSetAux_{i-1}$.
\begin{enumerate}[I.]
\item Suppose for a contradiction that $\Tuple{\Rule', \Subs'}$ is not applicable to $\AtomSetAux_{i-1}$.
\item By (9), (I), and Definition~\ref{definition:applicability}: there is some substitution $\SubsAux' \supseteq \Subs'$ with $\Head\SubsAux' \subseteq \AtomSetAux_{i-1}$.
\item By (7.b) and (II): $\Subs \subseteq (\Rewriting_{i-1} \circ \SubsAux')$.
\item By (2.b) and (II): $\Head (\Rewriting_{i-1} \circ \SubsAux') \subseteq \AtomSet_{i-1}$.
\item By (3.a) and Definition~\ref{definition:applicability}: $\Head\SubsAux \not\subseteq \AtomSet_{i-1}$ for all $\SubsAux \supseteq \Subs$.
\item (III), (IV), and (V) result in a contradiction and hence, we conclude that (i) does not hold.
\end{enumerate}
\item Let $\AtomSetAux_i$ be the minimal set containing the application of $\Tuple{\Rule', \Subs'}$ on $\AtomSetAux_{i-1}$ and $\{\EP(t, t) \mid t \in \EntitiesIn{\Terms}{\AtomSetAux_i}\}$.
\item Let $\Rewriting_i \supseteq \Rewriting_{i-1}$ be the rewriting such that $\Rewriting_i(f_w(\Subs'(\Vy))) = f_w(\Subs(\Vy))$ for all $w \in \Vw$.
\item By (2.a), (5), (a), and (b): the set $\AtomSetAux_i$ is a chase step of $\On'$  as it can be obtained by applying the tuple \Tuple{\Rule', \Subs'} to $\AtomSetAux_{i-1}$ and then exhaustively applying the rules of type \eqref{rule:reflexivity} in $\RuleSet'$.
\item By (2.b), (b), and (c): $\AtomSetAux_i\Rewriting_i \setminus \{\EP(t, u) \mid t, u \in \Terms\} = \AtomSet_i$.
\item By (T): the predicate \EP does not occur in $\Head$.
\item By (2.c), (b), (c), and (f): $\EP(t, u) \in \AtomSetAux_i$ for all $t, u \in \Terms$ with $\Rewriting_i(t) = \Rewriting_i(u)$.
\end{enumerate}
\item[E.] The rule $\Rule$ is an EGD.
That is, $\Rule$ is a formula of the form $\Body[\Vx] \to x \approx y$ with $x, y \in \Vx$.
We consider two possible cases depending on whether $\Tuple{\Rule', \Subs'}$ is applicable to $\AtomSetAux_{i-1}$.
\begin{itemize}
\item Assume that the tuple \Tuple{\Rule', \Subs'} is applicable to $\AtomSetAux_{i-1}$.
\begin{enumerate}[a.]
\item Let $\AtomSetAux_i$ be the atom set that results from applying $\Tuple{\Rule', \Subs'}$ to $\AtomSetAux_{i-1}$, and then exhaustively applying the rules of \eqref{rule:symmetry} and \eqref{rule:transitivity} as well as the rules of type \eqref{rule:replacement} in $\RuleSet'$.
\item Let $\Rewriting_i$ be the rewriting defined as follows.
\begin{enumerate}
\item If $\Subs(x) \prec \Subs(y)$, then $\Rewriting_i$ results from replacing every pair $\langle t, \Subs(y) \rangle \in \Rewriting_{i-1}$ with $\langle t, \Subs(x) \rangle \in \Rewriting_{i-1}$.
\item Otherwise, $\Rewriting_i$ results from replacing every pair $\langle t, \Subs(x) \rangle \in \Rewriting_{i-1}$ with $\langle t, \Subs(y) \rangle \in \Rewriting_{i-1}$.
\end{enumerate}
\item By (5), (2.a), (a), and Definition~\ref{definition:singularisation}: $\AtomSetAux_i$ is a chase step of $\On'$.
\item By (2.b), (a), and (b): $\AtomSetAux_i\Rewriting_i \setminus \{\EP(t, u) \mid t, u \in \Terms\} = \AtomSet_i$,
\item By (2.c), (a), and (b): $\EP(t, u) \in \AtomSetAux_{i-1}$ for all $t, u \in \Terms$ with $\Rewriting_{i-1}(t) =  \Rewriting_{i-1}(u)$.
\end{enumerate}
\item Assume that the tuple \Tuple{\Rule', \Subs'} is not applicable to $\AtomSetAux_{i-1}$.
\begin{enumerate}[a.]
\item Let $\AtomSetAux_i = \AtomSetAux_{i-1}$.
\item The rest of the proof is analogous to the previous case.
\end{enumerate}
\end{itemize}
\end{itemize}
\end{enumerate}
\end{proof}

\end{list}
\end{tr}

\end{document}